\documentclass[12pt]{amsart}
\usepackage{amsmath,amssymb,amsthm}
\usepackage[dvips]{graphicx}
\usepackage{color}
\usepackage{hyperref}
\usepackage{wasysym}
\newtheorem{lem}{Lemma}[section]

\newtheorem{thm}{Theorem}[section]

\theoremstyle{definition}

\theoremstyle{remark}

\theoremstyle{remark}
\newtheorem{remark}{Remark}[section]
\numberwithin{equation}{section}

\newcommand{\C}{{\mathbb C}}
\newcommand{\N}{{\mathbb N}}

\newcommand{\R}{{\mathbb R}}

\definecolor{blu}{rgb}{0,0,1}

\begin{document}
\title{Symmetry breaking for  Schr\"odinger-Poisson-Slater energy }
\author{Jacopo Bellazzini}
   \address{Universit\`a di Sassari \\Via Piandanna 4, 07100 Sassari, Italy}
 \email{jbellazzini@uniss.it}
   \author{Marco Ghimenti}
   \address{Dipartimento di Matematica  \\ Universit\`a di Pisa \\ Largo B. Pontecorvo 5, 56100 Pisa, Italy}
\email{ghimenti@mail.dm.unipi.it}

\maketitle
\begin{abstract}
We study the asymptotic behavior of  ground state energy for Schr\"odinger-Poisson-Slater energy functional. 
We show that ground state energy restricted to  radially
symmetric functions is above the ground state energy  when the number of particles is sufficiently large. 
\end{abstract}

The aim of this paper
is to show  symmetry breaking phenomena for the ground states energy of the following 
 Schr\"odinger-Poisson-Slater  (SPS) equation,
\begin{equation}\label{eq:main}
i\psi_{t}+ \Delta \psi - (|x|^{-1}*|\psi|^{2}) \psi+|\psi|^{\frac{2}{3}}\psi=0 \ \ \  \text{ in } \R^{3},
\end{equation}
where $\psi(x,t):\R^{3}\times[0,T)\rightarrow \C$ is the wave function, 
$*$ denotes the convolution,  and the nonlinear term $|\psi|^{\frac{2}{3}}\psi$ is the Slater correction term. \\
This  Schr\"odinger-type equation with a repulsive nonlocal Coulomb potential is obtained by approximation of the HartreeÐ-Fock equation describing a 
quantum mechanical system of many particles, see e.g.  \cite{BA,MA, BLSS}. As example, SPS equation gives a good appro"ximation of the time  evolution of an electron ensemble in a semiconductor crystal. In particular, in this mean field 
approximation the quantity $\int |\psi|^2dx=\rho^2$ describes the total amount of particles. 
The energy associated to a state $\psi$ is defined as 
$$E(\psi):= \frac 12   \|\nabla \psi \|_{L^2}^2 +\frac{1}{4}\int_{\mathbb{R}^3}\int_{\mathbb{R}^3}\frac{\left | \psi(x) \right |^2\left | \psi(y) \right |^2}{\left | x-y \right |}dxdy-\frac{3}{8}\int_{\R^3} |\psi|^{\frac{8}{3}}dx.$$
We define the \emph{ground state energy} the quantity $I_{\rho^2}$, defined as
\begin{equation*} \label{mini1}
I_{\rho^2}=\inf_{B_{\rho}} E(\psi)
\end{equation*}
where $ B_{\rho}:=\{ \psi  \in H^1(\R^3) \text{ such that } \|\psi \|_{2}^2=\rho^2\}$. 

Very briefly we summarize what is known about $I_{\rho^2}$ (see e.g \cite{BS}, \cite{BS1}, \cite{BJL}, \cite{GEPRVI}, \cite{SS}, \cite{CDSS})
\begin{itemize}
\item[(i)] $ -\infty<I_{\rho^2}<0$  for all  $\rho>0$
\item[(ii)] the ground state energy is weakly subadditive, namely that 
$$I_{\rho^2}\leq I_{\mu^2}+I_{\rho^2-\mu^2} \text{ for all } 0<\mu<\rho$$
\item[(iii)] the function $\rho \rightarrow I_{\rho^2}$ is continuous
\item[(iv)] $\lim_{\rho \rightarrow 0} \frac{I_{\rho^2}}{\rho^2}=0$
\item[(v)] minimizers for $E$ on $B_{\rho}$ exist when $\rho$ is sufficiently small
\item[(vi)] minimizers for $E$ on $B_{\rho}$ are radially symmetric when $\rho$ is sufficiently small
\end{itemize}

If we define 
\begin{equation*} \label{minirad1}
I_{\rho^2}^{rad}=\inf_{B_{\rho}^{rad}} E(\psi)
\end{equation*}
where $ B_{\rho}^{rad}:=\{ \psi  \in H^1(\R^3) \text{ radially symmetric such that } \|\psi \|_{2}^2=\rho^2\}$
it is not clear  if 
$$I_{\rho^2}^{rad}=I_{\rho^2} \text{ for all } \rho>0.  \ \ \ \ \ \ \ \ \ \ \ \ \ (Q)$$
We shall emphasize that for   Schr\"odinger-Poisson-Slater energy there are two terms that are in competition 
when we move from $u$ to $u^{\star}$, the symmetric decreasing rearrangement of $u$. Indeed, it is well understood that
for kinetic term $$\int_{\R^3}|\nabla u^{\star}|^2 dx\leq \int_{\R^3}|\nabla u|^2 dx$$
while for the Coulomb  term thanks to Riesz inequality 
 $$\int_{\mathbb{R}^3}\int_{\mathbb{R}^3}\frac{\left | u(x) \right |^2\left | u(y) \right |^2}{\left | x-y \right |}dxdy\leq \int_{\mathbb{R}^3}\int_{\mathbb{R}^3}\frac{\left | u^{\star}(x) \right |^2\left | u^{\star}(y) \right |^2}{\left | x-y \right |}dxdy$$
such that the  question $(Q)$ is still open.\\
Aim of this paper is  to give an  answer to question (Q). The fact that the energy functional as well as the $L^2$ constraint is rotationally invariant 
does not imply that the ground state energy for  radially symmetric functions coincides with ground state energy. Symmetry breaking can occur. This is exactly the case when
$\rho$ is sufficiently large.
\begin{thm}\label{mainthm}
There exists $\rho_0>0$ such that $I_{\rho^2}<I_{\rho^2}^{rad}$ for all $\rho>\rho_0.$
\end{thm}
\begin{remark}
Our result does not clarify if minimizers for $E$ restricted to $B_{\rho}$ exist for large $\rho$. This is still an open problem.
\end{remark}

Our argument is based on two ingredients.
The first  ingredient is the following scaling invariant inequality that holds only for radially symmetric function involving  the kinetic term, the Coulomb term and $L^p$ norms. 

\begin{equation}
\label{eq:dbound2}
\|\varphi\|_{L^{p}(\R^3)}\leq C(p,s) \|\varphi\|_{\dot H^{s}(\R^3)}^{\frac{\theta}{2-\theta}}\left(\iint_{\R^3\times \R^3}
\frac{|\varphi(x)|^2 |\varphi(y)|^2}{|x-y|} \, dxdy\right)^{\frac{1-\theta}{4-2\theta}}
\end{equation}
with $\theta=\frac{6-\frac{5}{2} p}{3-ps-p}$. Here the parameters $s$ and $p$ satisfy
$$ p\in (\frac{16s+2}{6s+1},  \frac 6{3-2s}]  \ \ \text{if}\ 1/2< s<3/2. $$

In our case we will use \eqref{eq:dbound2} in case $p=\frac{8}{3}$ and $s=1$ to control from above the Slater correction term  in terms of kinetic energy and  Coulomb energy.
In particular we obtain for radially symmetric functions that
\begin{equation}\label{eq:dbound3}
\|\varphi\|_{L^{\frac 83}(\R^3)}^{\frac 83}\leq C \|\nabla \varphi\|_{L^2(\R^3)}^{\frac 49}\left(\iint_{\R^3\times \R^3}
\frac{|\varphi(x)|^2 |\varphi(y)|^2}{|x-y|} \, dxdy\right)^{\frac 59}
\end{equation}
The importance of the previous inequality in the radial case is that it relates kinetic, Coulomb and Slater energy without information on the $L^2$ norm of the function.
 
The previous inequality  will imply that $\inf_{\rho >0}I_{\rho^2}^{rad}>-\infty.$ 

\begin{remark}
$L^p$ lower bounds in terms of kinetic and Coulomb energy are not a novelty (see e.g. \cite{BFV} for the general case). The first result concerning an improvement in radial case is due to Ruiz \cite{R} (see also \cite{R2})
using a weighted Sobolev embedding for radially symmetric functions due to \cite{SWW}. The scaling invariant radial bound \eqref{eq:dbound2} is proved in 
\cite{BGO} using a pointwise inequality found by De Napoli \cite{D} together with a lower bound for the Coulomb energy due to Ruiz \cite{R}.  
\end{remark}

The second ingredient is the asymptotic behaviour of the function 
\begin{equation}
\rho \rightarrow \frac{I_{\rho^2}}{\rho^2}
\end{equation}
for a translation invariant energy functional. We will show in Lemma \ref{lemgenn}
that the behaviour of $\rho \rightarrow \frac{I_{\rho^2}}{\rho^2}$ will imply that
$$I_{\rho^2}<I_{\mu^2} \text{ for all } 0<\mu<\rho$$
$$\inf_{\rho >0}I_{\rho^2}=-\infty$$
and hence our main result.
\begin{remark}
The importance of the function of $\rho \rightarrow \frac{I_{\rho^2}}{\rho^2}$ to show existence of minimizers is well established, see e.g. \cite{BS}. Indeed, the strict subadditivity 
inequality $$I_{\rho^2}< I_{\mu^2}+I_{\rho^2-\mu^2} \text{ for all } 0<\mu<\rho$$
which implies the existence of minimizers follows immediately from the strict monotonicity of the function of $\rho \rightarrow \frac{I_{\rho^2}}{\rho^2}$. In the radial case,
sufficient condition for the existence of minimizers is the strict monotonicity of  $\rho \rightarrow I_{\rho^2}$.
Here we show that the ratio between ground state energy and number of particles is crucial also for symmetry breaking.
\end{remark}

\begin{remark}
We notice that our argument is general and symmetry breaking appears if we substitute the energy functional with  
$$\tilde E(\psi):= \frac 12   \|\nabla \psi \|_{L^2}^2 +\frac{1}{4}\int_{\mathbb{R}^3}\int_{\mathbb{R}^3}\frac{\left | \psi(x) \right |^2\left | \psi(y) \right |^2}{\left | x-y \right |}dxdy-\frac{1}{p}\int_{\R^3} |\psi|^{p}dx$$
for $\frac{18}{7}<p<3.$ This fact follows again from the lower bounds for radially symmetric function \eqref{eq:dbound3}. The exponent $p=\frac{18}{7}$ is the left endpoint exponent for which \eqref{eq:dbound3} holds.
\end{remark}
\begin{remark}
We notice that $\lim_{\rho \rightarrow \infty} \frac{I_{\rho^2}^{rad}}{\rho^2}=0$. Indeed, arguing as Theorem 1.2 in \cite{R} or Theorem 2.2 in \cite{BGO}, there exist a family of functions
$u_{\rho} \in B_{\rho}^{rad}$ with $\rho \rightarrow \infty$ with kinetic and Coulomb terms uniformly bounded and vanishing Slater term. This fact implies that
$\sup_{\rho}I_{\rho^2}^{rad}<+\infty.$ 
\end{remark}
Our last result concerns   ground state energy in the radial case

\begin{thm}\label{thmrad}
The following alternative holds: either the function $\rho \rightarrow I_{\rho^2}^{rad}$ is monotone decreasing which implies the existence of  minimizers for $E(u)$ on $B_{\rho}^{rad}$  for all $\rho>0$, or there exist a static radial solution to \eqref{eq:main} belonging to $L^2(\R^3)$.
\end{thm}

\begin{remark}
The existence of static radial solution to \eqref{eq:main} has been proved for the first time by Ruiz \cite{R}. In \cite{BJL} it has been shown that real solution to
\begin{equation}\label{static}
-\Delta \psi + (|x|^{-1}*|\psi|^{2}) \psi-|\psi|^{p-2}\psi=0
\end{equation}
found by \cite{I} belong always to $H^1(\R^3)$ if $3\leq p<6$. It remains open the case $p=\frac{8}{3}$ discussed here.
\end{remark}
\section{Proof of Theorem \ref{mainthm}.}

\begin{lem}\label{scop}
Let $u\in B_{\rho^2}^{rad}$ then there exist $K_1>0$ not depending on $\rho$, such that
$$E(u)>-K_1.$$
\end{lem}
\begin{proof}
Let us call 
$$x= ||\nabla u||_{L^2(\R^3)}^2$$
$$y= \left(\iint_{\R^3\times \R^3}
\frac{|u(x)|^2 |u(y)|^2}{|x-y|} \, dxdy\right).
$$
From  inequality \eqref{eq:dbound3} one gets
$$E(u)>\frac 12 x+\frac 14  y  - \frac{C}{p} x ^{\frac 29} y^{\frac 59}$$
which immediately implies that $E(u)$ is bounded from below by a constant not depending on $\rho$.\\
In the general case, when we substitute the exponent $\frac 83$ for the Slater term with $\frac{18}{7}<p<3,$ we get
$$\tilde E(u)\geq \frac 12 x+\frac 14  y  - \frac{C}{p} x ^{\frac{5}{6}p-2} y^{1-\frac{p}{6}},$$
which again proves that the energy is bounded from below  by a constant not depending on $\rho$.

\end{proof}

We have the following general result for ground state energy of translation invariant energy functionals
\begin{lem}\label{lemgenn}
Let the function $\rho \rightarrow I_{\rho^2}$ fulfills the following conditions
\begin{equation}
-\infty<I_{\rho^2}<0  \text{ for all } \rho>0 \label{cond1}
\end{equation}
the ground state energy is weakly subadditive, namely that 
\begin{equation}\label{cond2}
I_{\rho^2}\leq I_{\mu^2}+I_{\rho^2-\mu^2} \text{ for all } 0<\mu<\rho 
\end{equation}
\begin{equation}\label{cond3} \text{ the function } \rho \rightarrow I_{\rho^2}  \text{ is continuous}
\end{equation}
\begin{equation} \label{cond4}
\lim_{\rho \rightarrow 0} \frac{I_{\rho^2}}{\rho^2}=0
\end{equation}
then 
$$I_{\rho^2}<I_{\mu^2} \text{ for all } 0<\mu<\rho$$
$$\inf_{\rho >0}I_{\rho^2}=-\infty$$
\end{lem}
\begin{proof}
The strict monotonicity of the function  $\rho \rightarrow I_{\rho^2}$ follows immediately from \eqref{cond1} and \eqref{cond2}.\\
Let we argue now by contradiction assuming that there exist $K_2>0$ such that $I_{\rho^2}>-K_2$ for all $\rho>0$. In this case we would have  $\lim_{\rho \rightarrow \infty } \frac{I_{\rho^2}}{\rho^2}=0.$ By conditions \eqref{cond1}, \eqref{cond3} and $\eqref{cond4}$ the function $\rho \rightarrow \frac{I_{\rho^2}}{\rho^2}$ attains a global minimum, i.e. there exists $\rho_0$ such that
$$\frac{I_{\rho_0^2}}{\rho_0^2}\leq \frac{I_{\rho^2}}{\rho^2} \text{ for all } \rho>0.$$
On the other hand, by the weak subadditivity condition \eqref{cond2},
$$ \frac{I_{2\rho_0^2}}{2\rho_0^2}\leq \frac{I_{\rho_0^2}+I_{\rho_0^2}}{2\rho_0^2}= \frac{I_{\rho_0^2}}{\rho_0^2}$$
which implies that   the function $\rho \rightarrow \frac{I_{\rho^2}}{\rho^2}$ attains a global minimum also at $\sqrt{2}\rho_0$.
The same argument shows that
$$ \frac{I_{k\rho_0^2}}{k\rho_0^2}= \frac{I_{\rho_0^2}}{\rho_0^2} \text{ for all }k \in \N$$
which implies that $\liminf_{\rho \rightarrow \infty} \frac{I_{\rho^2}}{\rho^2}<0$ that contradicts the fact that $\lim_{\rho \rightarrow \infty } \frac{I_{\rho^2}}{\rho^2}=0.$
\end{proof}
We show now that $I_{\rho^2}$ fulfills the assumptions of Lemma \ref{lemgenn}, see also \cite{BS}, \cite{SS}.
\begin{lem}
The function $\rho \rightarrow I_{\rho^2}$ fulfills the following conditions
$$-\infty<I_{\rho^2}<0  \text{ for all } \rho>0 $$
$$I_{\rho^2}\leq I_{\mu^2}+I_{\rho^2-\mu^2} \text{ for all } 0<\mu<\rho $$
$$\text{ the function } \rho \rightarrow I_{\rho^2}  \text{ is continuous}$$
$$\lim_{\rho \rightarrow 0} \frac{I_{\rho^2}}{\rho^2}=0$$
\end{lem}
\begin{proof}
For generality we substitute the exponent $\frac{8}{3}$ for the Slater term with $\frac{18}{7} <p<3$. The weak subadditive inequality is a general fact for translation invariant energy functionals, see e.g \cite{SS}.\\
We define
the following quantities:
$$A(u):=\int_{\R^{3}} |\nabla u|^{2}dx,\ \ \ B(u):=\iint_{\R^3\times \R^3}
\frac{|u(x)|^2 |u(y)|^2}{|x-y|} \, dxdy\,\ \ \ C(u):=-\int_{\R^{3}} |u|^{p}dx.$$
\\
\emph{Negativity}:\\
Let us 
consider, for $u\in B_{1}$, the rescaled function given by
$$u_{\theta,\beta}(x)=\theta^{1-\frac{3}{2}\beta}u(\frac{x}{\theta^{\beta}})$$ so that 
 $\bar \rho:=\|u_{\theta,\beta}\|_2=\theta $. 
  We easily find
the following scaling laws:
\begin{equation*}\label{scalingA}
A(u_{\theta,\beta})=
\theta^{2-2\beta}A(u),
\end{equation*}
\begin{equation*}\label{scalingN}
B(u_{\theta,\beta})=
\theta^{4-\beta}B(u),
\end{equation*}
\begin{equation*}\label{scalingp}
C(u_{\theta,\beta})=
\theta^{(1-\frac{3}{2}\beta)p+3\beta}C(u).
\end{equation*}
Notice that for $\beta=-2$ we get
$$E(u_{\theta,-2}) = \frac{\theta^6}{2} A(u)+\frac{\theta^{6}}{4}B(u)+\frac{\theta^{4p-6}}{p}C(u)$$ 
and we have  $4p-6<6$ since $p<3$. Hence for $\theta \rightarrow  0$ we have  $E(u_{\theta,-2})\rightarrow 0^-$.
This proves that there exists a small $\theta$, and therefore a small $\bar \rho$, such that
$$I_{s^{2}}<0 \ \ \forall\,s\in(0,\bar\rho].$$
By weak subadditivity  $I_{s^2}<0$ for all $s$.\\
\\
\emph{Continuity}:\\
We first prove that if $\rho_n \rightarrow \rho$ then $\lim_{n \rightarrow \infty}E_{\rho_n^2}=I_{\rho^2}.$ 
For every $n\in\mathbb N$, let $w_n \in B_{\rho_n}$ such that 
$E(w_n)<E_{\rho_n^2}+\frac{1}{n}<\frac{1}{n}$.  Therefore, by using the interpolation and the Sobolev inequality, we get
$$\frac{1}{2}\|\nabla w_{n}\|_{2}^{2}-C\rho_{n}^{\frac{6-p}{2}}\|\nabla w_{n}\|_{2}^{\frac{3(p-2)}{2}}\le \frac{1}{2}\|\nabla w_{n}\|_{2}^{2}-\frac{1}{p}\|w_{n}\|_{p}^{p}\le E(w_{n})<\frac{1}{n}.$$
Since $\frac{3(p-2)}{2}<2$ and $\{\rho_{n}\}$ is bounded, we deduce that
\begin{equation*}
\{w_{n}\} \ \text{ is bounded in } \ H^{1}(\R^{3}).
\end{equation*}
In particular $\{A(w_{n})\}$ and $ \{C(w_{n})\}$ are bounded sequences, and also
$ \{B(w_{n})\}$ since
$$\forall\, u\in H^{1}(\R^{3}): B(u)\le C \|u\|_{H^{1}(\R^{3})}^{4}.$$
So we easily find
\begin{align*}I_{\rho^2}&\leq E(\frac{\rho}{\rho_n}w_n)=\frac{1}{2}\left(\frac{\rho}{\rho_{n}}\right)^2A(w_{n})
+\frac{1}{4}\left(\frac{\rho}{\rho_{n}}\right)^{4}B(w_{n}) 
+\frac{1}{p}\left(\frac{\rho}{\rho_{n}}\right)^{p}C(w_{n})\\
&=E(w_n)+o(1)<I_{\rho_n^2}+o(1).\end{align*}
On the other hand, given 
 a minimizing sequence $\{v_n\}\subset B_{\rho}$ for $I_{\rho^{2}}$, 
we have
$$I_{\rho_n^2}\leq E(\frac{\rho_n}{\rho}v_n)=E(v_n)+o(1)=I_{\rho^{2}}+o(1).$$
We get $\lim_{n \rightarrow \infty}I_{\rho_n^2}=I_{\rho^2}.$ \\
\\
\emph{Asymptotic behavior at zero}:\\
In order to show that
$\lim_{\rho \rightarrow 0}\frac{I_{\rho^2}}{\rho^2}=0$, we notice that
$$\frac{G_{\rho^2}}{\rho^2} \leq\frac{I_{\rho^2}}{\rho^2}<0$$
where $G_{\rho^2}$ is defined by 
\begin{equation} \label{Grho}
G_{\rho^{2}}=\inf_{B_{\rho}} \ G(u)
\end{equation}
where
$$G(u)= \frac 12   \| \nabla u \|_{L^2(\R^3)}^2 -\frac{1}{p}\int_{\R^{3}} |u|^pdx.
$$
It is well established that minimizers for $G_{\rho^{2}}$ exist for all $\rho$ (indeed it is not difficult to show that the function
$\rho \rightarrow \frac{G_{\rho^{2}}}{\rho^2}$ is monotone decreasing). Let us call $\bar u\in B_{\rho^2}$  the minimizer for $G_{\rho^{2}}$, i.e the unique positive solution to 
$$-\Delta u -\lambda_{\rho} u -|u|^{p-2}u=0$$
where $\lambda_{\rho}<0$ is the corresponding Lagrange multiplier.
By Pohozaev identity it follows that
$$E(\bar u)= G_{\rho^{2}}=c \lambda_{\rho} \rho^{2}.$$
To conclude we shall prove that $\lim_{\rho \rightarrow 0} \lambda_{\rho}=0.$Taken $\tilde u$ the unique positive solution to
$$-\Delta u +u -|u|^{p-2}u=0$$
we have by scaling that $\bar u=(-\lambda_{\rho})^{\frac{1}{p-2}} \tilde u(\sqrt{(-\lambda_{\rho})}x)$ and hence
$$\rho^2=||\bar u||_{L^2(\R^3)}^2=(-\lambda_{\rho})^{\frac{2}{p-2}-\frac 32}||\tilde u||_{L^2(\R^3)}^2$$
and noticing that $\frac{2}{p-2}-\frac 32>0$ for $p<\frac{10}{3}$ we get the required estimate.
\end{proof}

\begin{proof}[Proof of Theorem \ref{mainthm}]$$$$
By Lemma \ref{lemgenn} the ground state energy $I_{\rho^2}$ is  strictly decreasing as a function of $\rho$ and  $\inf_{\rho >0}I_{\rho^2}=-\infty$. On the other
hand by Lemma \ref{scop}, for radially symmetric functions $\inf_{\rho >0}I_{\rho^2}^{rad}>-\infty$  and therefore symmetry breaking occurs for sufficiently large $\rho.$
\end{proof}

\section{Proof of Theorem \ref{thmrad}.}
The fact that the strict monotonicity of $\rho \rightarrow I_{\rho^2}^{rad}$ is sufficient for proving the existence of minimizers in the radial case follows immediately
from the weak continuity of the Coulomb energy and of Slater energy.
Indeed, taken a minimizing sequence $u_n\in B_{\rho^2}^{rad}$, $u_n\rightharpoonup u\neq 0$. Let us assume by contradiction that $u\in B_{\mu^2}^{rad}$ with $0<\mu<\rho$, such that
$$I_{\mu^2}^{rad} +o(1) \leq E(u)+o(1)=E(u_n)=I_{\rho^2}^{rad}+o(1).$$
On the other hand, $I_{\mu^2}^{rad}>I_{\rho^2}^{rad}$ and therefore $\mu=\rho$ and $u$ is a minimizer for $E$ on $B_{\rho^2}^{rad}$.\\
It has been proved in \cite{GEPRVI} that $I_{\rho^2}=I_{\rho^2}^{rad}$ for sufficiently small $\rho.$ On the other hand, by the weak subadditivity inequality it is clear that
$I_{\rho^2}<I_{\mu^2}$ for $0<\mu<\rho$. This fact proves that for sufficiently small $\rho$, $I_{\rho^2}^{rad}<I_{\mu^2}^{rad}$.\\
Let us define  $c:=\min_{[0, \rho]}I_{s^2}^{rad}<0$ and 
$$\rho_{0}:=\min \left\{ s \in [0, \rho] \text{ s.t } I_{s^2}^{rad}=c \right\}.$$
It is clear that  $\rho_0>0$ and 
\begin{equation}\label{minimo}
\forall\,s\in[0,\rho_{0}): I_{\rho_{0}^{2}}^{rad}<I_{s^{2}}^{rad}
\end{equation}
namely, the function $[0, \rho_0]\ni s\mapsto I_{s^2}\in \R_{-}$ achieves the minimum only in  $s=\rho_0,$
by definition of $\rho_{0}.$ From this fact we deduce the existence of  $\bar u\in B_{\rho_{0}}$
 such that $E(\bar u)=I_{\rho^{2}_0}^{rad},$  and therefore $\bar u$ fulfills
\begin{equation}\label{staz}
-\Delta u + (|x|^{-1}*|u|^{2}) u-|u|^{\frac 23 }u =0
\end{equation}
 where $\lambda$ is the associated Lagrange multiplier.
 \\
 Now let us suppose that $\rho_{0}<\rho,$ 
 which implies the existence of a  sequence $\theta_n>1$  with $\theta_n \rightarrow 1$ such that
 $$E(\theta_n \bar u)\geq E(\bar u) \text{ with } \theta_n \rightarrow 1$$
 and 
 $$E(\theta  \bar u)\geq E(\bar u) \text{ with } 0<\theta<1.$$
 This fact implies that
 $$\frac{d}{d \theta}E(\theta \bar u)|_{\theta=1}=0,$$
and therefore the identity
\begin{equation}\label{imp}
\int_{\R^{3}} |\nabla \bar u|^{2}dx+ \iint_{\R^3\times \R^3}
\frac{|\bar u(x)|^2 |\bar u(y)|^2}{|x-y|} \, dxdy-\int_{\R^{3}} |\bar u|^{\frac{8}{3}}dx=0.
\end{equation}
From \eqref{imp} we notice that the associated Lagrange multiplier in \eqref{staz} is $\lambda=0$,
and hence $\bar u$ solves
$$-\Delta u + (|x|^{-1}*|u|^{2}) u-|u|^{\frac 23}u=0,$$
proving that  $\bar u$ is a static solution to \eqref{eq:main} with $\bar u \in B_{\rho_0^2}^{rad}$.


\begin{thebibliography}{99}
\bibitem{BA} {C. Bardos, F. Golse, A. D. Gottlieb, N. Mauser,} Mean field dynamics of fermions and the time-dependent Hartree-Fock equation, J. Math. Pures Appl. (9) 82 (2003), no. 6, 665-683.
\bibitem{BFV} {J.Bellazzini, R.L. Frank, N. Visciglia, } Maximizers for Gagliardo-Nirenberg inequalities and related non-local problems,  Math. Annalen  (2014), no. 3-4, 653Ð-673.
\bibitem{BGO} J. Bellazzini, M. Ghimenti, T. Ozawa, {\sl Sharp lower bounds for Coulomb energy},  Math. Res. Lett. ( in press ) arXiv:1410.0598 (2015)
\bibitem{BJL} J. Bellazzini, L. Jeanjean, T. Luo, Existence and instability of standing waves  with prescribed norm for a class of Schr\"odinger-Poisson  equations, Proc. London Math. Soc. (3) 107,(2013), 303Ð-339.
\bibitem{BS1} {J. Bellazzini, G. Siciliano,}  Stable standing waves for a class of nonlinear
Schr\"odinger-Poisson equations, Z. Angew. Math. Phys., 62 (2011), no. 2, 267-280.
\bibitem{BS} {J. Bellazzini, G. Siciliano,} Scaling properties of functionals and existence of constrained minimizers, J. Funct. Analysis, 261 (2011), no. 9, 2486-2507.
\bibitem{BLSS} O. Bokanowski, J.L. Lopez, O.  Sanchez, J. Soler,  Long time behaviour to the Schršdinger-Poisson-$X^{\alpha}$ systems. Mathematical physics of quantum mechanics, 217Ð-232, Lecture Notes in Phys., 690, (2006)
\bibitem{CDSS}I. Catto, J.  Dolbeault, O. Sanchez, J. Soler,
Existence of steady states for the Maxwell-Schrodinger-Poisson system: exploring the applicability of the concentration-compactness principle,
Math. Models Methods Appl. Sci. 23 (2013), no. 10, 1915Ð-1938. 
\bibitem{D}  P.L. De N\'{a}poli, Symmetry breaking for an elliptic equation involving the Fractional
 Laplacian, arXiv:1409.7421 
\bibitem{GEPRVI} {V. Georgiev, F. Prinari,  N. Visciglia,} On the radiality of constrained minimizers to the Schr\"odinger-Poisson-Slater energy,  Ann. Inst. H. Poincar\'e Anal. Non Lin\'eaire (2012), no. 3, 369Ð-376.
\bibitem{I}  I. Ianni,  D. Ruiz, Ground and bound states for a static Schr\"odinger-Poisson-Slater problem. Commun. Contemp. Math. 14 (2012), no. 1, 1250003
\bibitem{MA} {N. J. Mauser}, The Schr\"odinger-Poisson-X$\alpha$ equation, Appl. Math. Lett. 14 (2001), no. 6, 759-763.
\bibitem{R2} D. Ruiz, The Schr\"odinger-ÐPoisson equation under the effect of a nonlinear local term, J. Funct. Anal.
237 (2006) 655Ð-674.
\bibitem{R} D. Ruiz, On the Schr\"odinger-Poisson-Slater system: behavior of minimizers, radial and nonradial cases,  Arch. Ration. Mech. Anal. 198 (2010), no. 1, 349Ð-368
\bibitem{SS} {O. Sanchez, J. Soler,} Long-time dynamics of the Schr\"odinger-Poisson-Slater system, J. Statist. Phys., 114 (2004), no. 1-2, 179-204.
\bibitem{SWW} J. Su, Z.-Q. Wang, M. Willem,Weighted Sobolev embedding with unbounded
and decaying radial potentials, J. Differential Equations 238 (2007), 201-219. 

\end{thebibliography}
\end{document}